\documentclass[12pt]{article}
\usepackage{epsf}
\usepackage{epsfig}
\usepackage{graphicx}
\usepackage{color}
\usepackage{amsmath}
\usepackage{mathrsfs}

\usepackage{amsthm}
\usepackage{latexsym}
\usepackage{amssymb}
\usepackage{amsmath}

\usepackage{stmaryrd}

\usepackage{epsf}
\usepackage{verbatim}

\newtheorem{theorem}{Theorem}[section]

\newtheorem{lemma}{Lemma}[section]

\newtheorem{definition}{Definition}[section]
\newtheorem{proposition}{Proposition}[section]

\begin{document}
\title{Base collapse of holographic algorithms}

\author{Mingji Xia\thanks{State Key Laboratory of Computer Science,
 Institute of Software, Chinese Academy of Sciences, Beijing, P. R.
 China.  {\tt mingji@ios.ac.cn} }}


\vspace{0.3in}
\maketitle

\begin{abstract}
A holographic algorithm solves a problem in domain of size $n$, by reducing it to counting perfect matchings in planar graphs. It may simulate a $n$-value variable by a bunch of $t$ matchgate bits, which has $2^t$ values. The transformation in the simulation can be expressed as a $n \times 2^t$ matrix $M$, called the base of the holographic algorithm. We wonder whether more matchgate bits bring us more powerful holographic algorithms.
In another word, whether we can solve the same original problem, with a collapsed base of size $n \times 2^{r}$, where $r<t$.

Base collapse was discovered for small domain $n=2,3,4$. For $n=3, 4$, the base collapse was proved under the condition that there is a full rank generator. We prove for any $n$, the base collapse to a $r\leq \lfloor \log n \rfloor$, with some similar conditions. One of them is that the original problem is defined by one symmetric function. In the proof, we utilize elementary matchgate transformations instead of matchgate identities.


\end{abstract}

%
%
%
%
%
%
%
%

\section{Introduction}

Holographic algorithm \cite{Valiant08-HA} is a method of designing polynomial time algorithm for counting problems. It usually solves counting problems in planar graphs by reducing it to counting perfect matchings in planar graphs, through holographic reductions.

A perfect matching of a graph, is a subset of edges such that each vertex appears exactly once, as the endpoint of these edges. The edges may have weights, and the weight of a perfect matching is the product of its edge weights.
Given a planar graph, the summation of its perfect matchings weights can be computed in polynomial time \cite{TemperleyFisher1961,Kasteleyn1961,Kasteleyn1967}.
We denote this problem by \#Pl-PerfMatch. We can design gadgets in \#Pl-PerfMatch, called matchgate \cite{Valiant02-QC,Valiant08-HA} . A holographic algorithm is designed by finding proper matchgates, and then applying a proper local transformation, called holographic reduction.

Holographic reduction \cite{Valiant08-HA,DBLP:conf/icalp/Valiant05,DBLP:conf/focs/Valiant06}
gives an important equivalence relation of counting problems. It transfers many properties between equivalent problems, including polynomial time algorithm \cite{DBLP:conf/icalp/Xia10}, exponential time algorithm, \#P-hardness \cite{DBLP:journals/tcs/XiaZZ07}. Its applications are not restricted to planar problems, for example, Fibonacci gates \cite{DBLP:conf/focs/CaiLX08} and counting graph homomorphisms \cite{DBLP:conf/icalp/Xia10}. In many complexity dichotomy theorems for counting problems in general graphs \cite{DBLP:conf/soda/CaiLX11, DBLP:journals/siamcomp/CaiLX11,DBLP:conf/soda/CaiLX13} , it is a very important method for both hardness and tractability. In this paper, to avoid confusion, we do not call these algorithms utilizing holographic reduction in general graphs holographic algorithm. Holographic algorithm means a problem is solved by reducing to \#Pl-PerfMatch.

After getting the complexity of a set of counting problems in general graphs, sometimes the set of planar version problems is studied. It is interesting that usually the new presented tractable problems are all solved by holographic algorithms \cite{DBLP:conf/focs/CaiLX10}, according to several pairs of dichotomy theorems. People may guess, that under holographic reduction, \#Pl-PerfMatch is the canonical form of all counting problems which is hard in general graphs but tractable in planar graphs. However, recently a new counting algorithm for planar graph appears \cite{DBLP:journals/corr/CaiFGW15}, breaking this guess.

We introduce holographic reduction and algorithm with a bit more intuitive details, starting from the problem definitions.

The well known SAT problem is one of the problems in the CSP family (Constraint Satisfaction  Problems). As other CSP problems, its instance can be drawn as a bipartite graph. Variable vertices are located at the left side, and the constraint vertices are located at the right side. A variable may appear in several constraints, connected by its edges.
The constraints of SAT are disjunction relations affected by negations of inputs. Let $\mathbf{D}$ denote these available relations. SAT is CSP$(\mathbf{D})$.

In fact, we can look a variable vertex also as a constraint, which is an equality relation. At the same time, its edges are looked as variables, which must take the same value as required by this equality relation. Let $\mathbf{E}$ denote the set of equality relations. SAT may be denoted as $\mathbf{E}|\mathbf{D}$. To define general problems, we may use any set as the available relation set for the vertices on the left side, as well as the set for the right side. In this paper, we always face the counting version of this kind of problems, with bipartite instances, described by two available function sets.

Holographic reduction was born with this kind of bipartite instances.
Suppose the original problem is  $\#\mathbf{F}|\mathbf{T}$ where $\mathbf{F}=\{F\}$ and $\mathbf{T}=\{T\}$ and $F,T$ are functions in variables taking $n$ discrete values. An example of its instance is shown as the first graph in Figure \ref{Fig HRA}.

\begin{figure}[hbtp]
	\begin{center}
		\includegraphics[width=\textwidth]{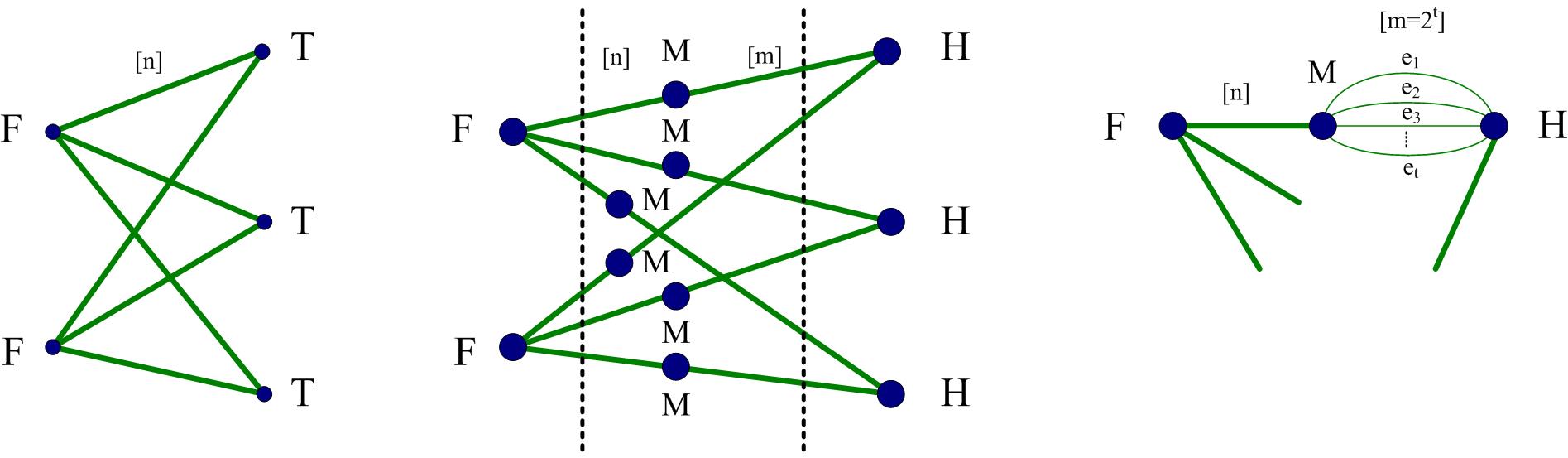}
	\caption{Holographic reduction and algorithm.}
	\label{Fig HRA}
	\end{center}
\end{figure}

The base $M$ of holographic reduction is a matrix of size $n \times m$. It can transform $H$ into $T$, where $H$ is a function in two variables of $m$ values. That is, as a gadget, a $H$ connected by two $M$s simulates a $T$. Let $\mathbf{H}=\{H\}$. We denote this transformation by $M \mathbf{H}=\mathbf{T}$. The gadget here is inherently similar to the gadgets used in reductions for proving NP-hardness.
We replace each $T$ by a gadget to get the second graph in Figure \ref{Fig HRA}.

If we separate the graph from the left dashed line, it is still $\#\mathbf{F}|\mathbf{T}$, also $\#\mathbf{F}|M\mathbf{H}$. If we sperate from the right dashed line, it is $\# \mathbf{F} M |\mathbf{H}$. Holographic reduction says $\#\mathbf{F}|M\mathbf{H}$ and $\# \mathbf{F} M |\mathbf{H}$ are equal.

In a holographic algorithm, the second problem $\# \mathbf{F} M |\mathbf{H}$  is solved by \#Pl-PerfMatch. Hence, all instances are planar graphs. Both $F M^{\otimes 3}$ and $H$ are simulated by matchgates. The matchgate realizing $F M^{\otimes 3}$ is called generator, and the one realizing $H$ is called recognizer.

Usually, in a holographic algorithm, $n$ is small and $m$ is $2$. However, there is no nothing forbidding to use a bunch edges as one input variable of $H$. For example, we may set $m=2^t$, and use $t$ edges to simulate a variable of $m$ values as the third graph  in Figure \ref{Fig HRA}.

To get a holographic algorithm, both sides are designed simultaneously. On one hand, we need to construct a matchgate $H$, such that $M$ transforms it into $T$. On the other hand, we also need $M$ to transform $F$ into a matchgate. There are characterizations about what's kinds of functions can be realized by matchgates and transformed matchgates \cite{DBLP:journals/ijsi/CaiC07,CCL2,DBLP:journals/jcss/CaiL11,DBLP:journals/mst/CaiL10,DBLP:journals/tcs/CaiL10}.

The functions that can be realized by matchgates must satisfy a system of equations called matchgate identities \cite{Valiant02-QC,CC1,CCL2,DBLP:journals/toc/CaiG14}. According to these identities, a matchgate function  has only polynomial many free values, and all other exponential many values are decided by them. It seems that we need a huge $t$ to simulate an arbitrary function.

However, in this paper, we prove $r \leq \lfloor \log n \rfloor$ bits are enough, under some conditions.
Because of the complicated requirements to design a holographic algorithm, we can not simulate arbitrary functions.
If we look the base $M$ as a signal channel, $ \lceil\log n \rceil$ is the lower bound to transfer a signal from the left side to the right side without any loss.
Because of the restriction of matchgates, $F$ can only utilize a part of $M$. If $n$ is not a power of $2$, the proof implicates at least $\log n- \lfloor \log n \rfloor$ bits of the channel are wasted.

When $n=2$, base collapse is proved in \cite{DBLP:conf/coco/CaiL07,DBLP:journals/tcs/CaiL09}.
When $n=3,4$, base collapse is proved in \cite{CaiFu}, under a condition about the existence of a full rank generator. These pioneer researches hint there may be a general base collapse.
The method in this paper is different from \cite{DBLP:journals/tcs/CaiL09,CaiFu}, which analyzes matchgate functions through matchgate identities. It inherits the matchgate transformation technique. This technique is firstly used in \cite{LX08stacs}, to prove the group property of nonsingular matchgate functions and that 2 bits matchgates are universal for matchcircuits.  In this paper, it is applied to not only invertible matchgates but also general matchgates. We get a deeper insight of the transformations, also a byproduct about the rank of matchgates.

The matchgate transformations are introduced in Section \ref{sec ele}.
In Section \ref{sec cnanonical form}, we use them to simplify an arbitrary matchgate into a canonical form.
The main theorem about the base collapse is proved in Section \ref{sec collapse}.
It is unexpected that the transformations specialized  for matchgate functions can deduce the property of the base of holographic algorithms, which is not necessary a matchgate function.

\section{Preliminary}

\subsection{Counting problem}

Let $[n]$ denote the set $\{0,1,\ldots,n-1\}$. Suppose $\mathbf{F}$ and $\mathbf{H}$ are two sets of functions in variables of domain $[n]$. We define a counting problem $\#\mathbf{F}|\mathbf{H}$. We call this kind of counting problems  {\bf \#BCSP} (\#Bi-restriction Constraint Satisfaction Problem). Intuitively, given a bipartite graph as its instance, we look each edge as a variable, and a vertex as a function in its incidental edges. A vertex on the left (resp. right) side must pick functions in $\mathbf{F}$ (resp. $\mathbf{H}$).  The answer is the summation of the product of these vertex functions, over all assignments to the edge variables. In the strict definition, we use mapping $\phi$ to specify which function is associated to each vertex, and use mapping $\psi$ to specify how a vertex function is applied to the vertex's edges.

An instance $(G,\phi_l,\phi_r)$ of $\#\mathbf{F}|\mathbf{H}$ is a bipartite graph
$G(U,V,E)$ and two mappings, $\phi_l:$ $U \rightarrow \mathbf{F}$ and $\phi_r:$ $V \rightarrow
\mathbf{H}$. Let $F_v$ denote the value of $\phi_l$ or $\phi_r$
on $v$. $\phi_l$ and $\phi_r$ satisfy that the arity of $F_v$ is $d_v$, the degree of $v$.
The bipartite graph $G$ is given as two one to one mappings,
$\psi_l:$ $(v,i) \rightarrow e$ and $\psi_r:$ $(u,i) \rightarrow e$,
where $v \in V$ and $u\in U$ respectively, $i \in [d_v]$ (resp.
$[d_u]$), and $e \in E$ is one of the edges incident to $v$ (resp.
$u$). Let $e_{v,i}=(\psi_l \cup \psi_r)(v,i)$, $v\in U \cup V$.


The value on this instance is defined as a summation over all assignments $\sigma$ of edges,
\[\#\mathbf{F}|\mathbf{H}(G,\phi_l,\phi_r)= \sum_{\sigma : E \rightarrow [n]}
\quad \prod_{v \in U \cup V}
F_v(\sigma(e_{v,1}),\sigma(e_{v,2}),\ldots,\sigma(e_{v,d_v}))
\textrm{.}\]


The two sets $\mathbf{F}$ and $\mathbf{H}$ tell the form of available functions that can be used as $F_v$ in a \#BSCP problem $\#\mathbf{F}|\mathbf{H}$. Since they do not affect how the value of an instance is defined, sometimes we use notation $\#(G,\phi_l,\phi_r)$ or $\#G$ instead of $\#\mathbf{F}|\mathbf{H}(G,\phi_l,\phi_r)$.

\subsection{Gadget}

A  gadget is a triple $\Gamma=(G,\phi,\psi)$, where $G(V,E,D)$ is a graph with vertex set $V$, common edge set $E$ and dangling edge set $D=\{x_1,x_2,\ldots,x_d\}$. A dangling edge $x_i$ is an edge containing only one endpoint. Formally, $E \subseteq V \times V$ and $D \subseteq V$. They are also called, internal edge and external edge. Mapping $\phi$ assigns each vertex $v$ a function $F_v$ of degree $d_v$. Mapping $\psi$ maps $(v,i), v \in V, i \in [d_v]$ to the internal and external edges incident to $v$, and $\psi(v,i)$ is denoted by $e_{v,i}$. That is, $\psi$ gives an ordering of $v$'s edges, such that they can be fed to $F_v$ without ambiguousness.

The function of a gadget $\Gamma$ is $F_\Gamma: [n]^d \rightarrow \mathcal{C}$.
Given $\tau: D \rightarrow [n]$, together with a $\sigma: E  \rightarrow [n]$ we have an assignment $\sigma \cup \tau$ to all edges.

\[F_\Gamma(\tau(x_1), \tau(x_2),\ldots, \tau(x_d))= \sum_{\sigma : E \rightarrow [n]}
\quad \prod_{v \in  V}
F_v((\sigma \cup \tau)(e_{v,1}),(\sigma \cup \tau)(e_{v,2}),\ldots,(\sigma \cup \tau)(e_{v,d_v}))
\textrm{.}\]

If we omit the notations $\tau$ and $\sigma$, and abuse the names of the edges, it is just $F_\Gamma(x_1,x_2,\ldots,x_d)=\sum_{e}  \prod_{v \in V}F_v(e_{v,1},e_{v,2},\ldots,e_{v,d_v})$.

An instance of $\#\mathbf{F}|\mathbf{H}$ is a gadget. It has arity $d=0$, and only $n^d$ value is defined for it.
A gadget of a specific problem must respect the requirement of that problem. A gadget of a $\#\mathbf{F}|\mathbf{H}$ problem, can be looked as a part of its instance.

The gadget with graph $G(\{v\}, E=\phi ,D=\{x_1,x_2,\ldots,x_d\}$) is very simple.
If $F_v$ is equal to the function $F_\Gamma$ of a gadget $\Gamma$, we can substitute each by the other in any instance or gadget. This property was used widely in reductions for both counting problems and decision problems, where the relation of a gadget is defined by $\bigvee \bigwedge$ instead of $\sum\prod$. This property and the associative law introduced later, are two statements from two aspects of the same thing. We omit the proof, which is straightforward from the definition.

\begin{figure}[hbtp]
	\begin{center}
		\includegraphics[width=0.5\textwidth]{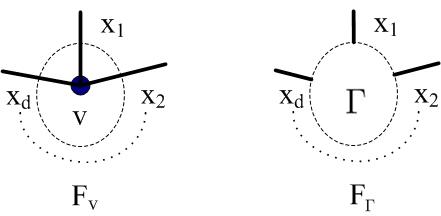}
	\caption{If $F_v=F_\Gamma$, they have the same utility.}
	\label{}
	\end{center}
\end{figure}

Suppose a binary function $F$ in $x_1,x_2$ is associated to a vertex of degree $2$. The $n^2$ values of $F$ can be expressed in many forms. Usually, we look the first gadget of Figure \ref{Fig vector and matrix} as the row vector form $F_{x_1x_2}$  indexed by $x_1x_2 \in [n]^2$. The second gadget shows a matrix form $(F_{x_1,x_2})$  indexed by row index $x_1$ and column index $x_2$. The third gadget shows the transpose of $(F_{x_1,x_2})$, $(F_{x_1,x_2})'=(F_{x_2,x_1})$  indexed by row index $x_2$ and column index $x_1$. The expressions works for functions of higher arities. Just imagine the edge in the pictures as a bunch of edges.

\begin{figure}[hbtp]
	\begin{center}
		\includegraphics[width=0.7\textwidth]{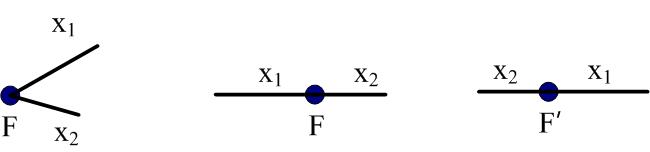}
	\caption{Different forms of a binary function $F$.}
	\label{Fig vector and matrix}
	\end{center}
\end{figure}

The functions of some gadgets coincide with vector matrix multiplications. For example, in Figure \ref{Fig multiplication},
\[{F_\Gamma}_{\; x_1}=F_{e_1} H_{e_1,x_1} \text { , and} \]
\[{F_\Delta}_{\; x_1,x_2x_3}=F_{x_1,e_1e_2} H_{e_1e_2,x_2x_3}.\]

\begin{figure}[htb]
	\begin{center}
		\includegraphics[width=0.5\textwidth]{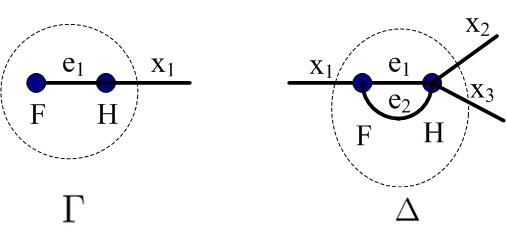}
	\caption{Examples of vector and matrix multiplications.}
	\label{Fig multiplication}
	\end{center}
\end{figure}

Tensor product is also a special composition in matchgates. For examples, in Figure \ref{Fig tensor},
\[{F_\Gamma}_{\; x_1x_3,x_2x_4}=M_{x_1,x_2} \otimes  N_{x_3,x_4}. \]

We get the second gadget in two steps. Firstly, we combine $M$ and $N$ and the empty on the edge $x_3$ to get a matrix
$H_{e_1e_2x_3,x_1x_2x_3}=M_{e_1,x_1} \otimes  N_{e_2,x_2} \otimes E_{x_3,x_3}$, where $E$ stands for the identity matrix.
Secondly, we multiply $F$ and $H$ as a row vector multiplying with a matrix. Put together,
\[{F_\Delta}_{\; x_1x_2x_3}=F_{e_1e_2x_3} (M_{e_1,x_1} \otimes  N_{e_2,x_2} \otimes E).\]

\begin{figure}[htb]
	\begin{center}
		\includegraphics[width=0.5\textwidth]{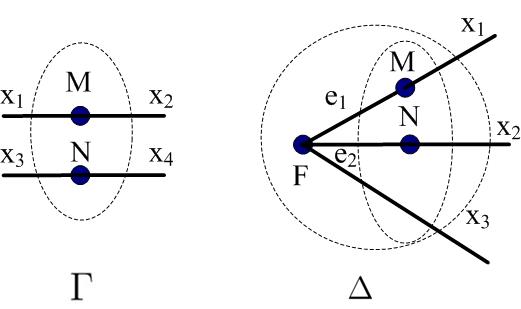}
	\caption{Tensor product.}
	\label{Fig tensor}
	\end{center}
\end{figure}

We denote $M \times M$ by $M^{\otimes 2}$. Generally, $M^{\otimes d+1}=M^{\otimes d} \otimes M$.

Both tensor product and matrix multiplication obeys associative law. Generally, in any gadget, we can compute any part firstly and reach the same final result.

\subsection{Holant theorem}

We say two problems  $\#\mathbf{F}|\mathbf{H}$ and
$\#\mathbf{P}|\mathbf{Q}$ are {\it result equivalent}, if there are
two bijections $\sigma_l:$ $ \mathbf{P}\rightarrow \mathbf{F}$ and
 $\sigma_r:$ $ \mathbf{Q}\rightarrow \mathbf{H}$, such that
for any instance $(G,\phi_l,\phi_r)$,
\[ \#\mathbf{F}|\mathbf{H}(G,\phi_l,\phi_r)=  \#\mathbf{P}|\mathbf{Q}(G,\sigma_l \circ \phi_l,\sigma_r \circ\phi_r).\]

Let $M \mathbf{H}=\{M^{\otimes R_H}H | H \in \mathbf{H}\}$, where $R_H$ denotes the arity of $H$. Similarly, let
$\mathbf{F} M=\{F M^{\otimes R_F} | F \in \mathbf{F}\}$.

\begin{theorem} [Holant theorem \cite{Valiant08-HA}]  \label{thm:holant}
Suppose $F$ is a function over $[n]^s$, and $H$ is a function over
$[m]^t$, and $M$ is an $n \times m$ matrix. Problems
$\#\{F\}|\{M^{\otimes t}H\}$ and $\#\{F M^{\otimes s}\}|\{H\}$
 are result equivalent. Generally, $\#\mathbf{F}|M\mathbf{H}$ and $\#\mathbf{F}M|\mathbf{H}$
 are result equivalent, under the proper bijections.
\end{theorem}

We illustrate Holant theorem by an instance $G$ as shown in the first graph of Figure \ref{Fig HRA}. In this instance, $F$ has arity $3$ and $H$ has arity $2$. In the second picture, If cut through the left dash line, we get problem $\#\{F\}|\{M^{\otimes 2}H\}$, while if cut through the right dash line, we get problem $\#\{F M^{\otimes 3}\}|\{H\}$. By associative law, both problems give the same value $\#G$.

\subsection{Planar counting problem}

We can define planar \#BCSP problem \#Pl-$\mathbf{F}|\mathbf{H}$ and its gadget similarly, with the following differences.

An instance or a gadget contains also a planar embedding of $G$, such that all external edges dangling on the outer face. Mapping $\phi$ must order the edges of $v$ such that in the planar embedding when going around $v$ clockwise or anticlockwise, starting from $e_{v,1}$, we meet them in the order $e_{v,1},e_{v,2},\ldots,e_{v,d_v}$.
In any gadget, the order of external edges $x_1,x_2,\ldots,x_d$ satisfies that when going along the outer face of $\Gamma$ anticlockwise, starting from $x_1$, we meet dangling edges in the order $x_1,x_2,\ldots,x_d$.

\begin{definition} \label{def  C}
Call the following function $C$ in Boolean variables $x_1,x_2,x_3,x_4$ sign crossover function.
$C(0000)=C(0101)=C(1010)=1$, $C(1111)=-1$, and $C$ is $0$ on other inputs.
In matrix form $(C_{x_1x_2,x_4x_3})=\left ( \begin{array}{cccc} 1 & 0 & 0 & 0  \\ 0 & 0 & 1 & 0 \\ 0 & 1 & 0 & 0  \\ 0 & 0 & 0 & -1  \end{array} \right )$.
\end{definition}

In a \#Pl-$\mathbf{F}|\mathbf{H}$ problem, the order of variables becomes important.
For example, let $C$ is a function in $\mathbf{F}$.  $C'(a_2,a_1,a_3,a_4)=C(a_1,a_2,a_3,a_4)$ is another function, which is not necessary available at left side in \#Pl-$\mathbf{F}|\mathbf{H}$, although in the non-planar version $\#\mathbf{F}|\mathbf{H}$ problem, we can simulate $C'$ by $C$ easily.

The Holant theorem also holds for planar \#BCSP problems.

\subsection{Matchgate}

Matchgate is defined and used in two contexts, matchcircuits \cite{Valiant02-QC}  and  holographic algorithms \cite{Valiant08-HA}. It is gadget in \#Pl-PerfMatch problem. The function of a matchgate is also called signature.
Sometimes, we do not distinguish a matchgate $\Gamma$ and its function $F_\Gamma$, but we know that a function may be realized by different matchgates.

We use $[F_0, F_1, \cdots, F_d]$ to denote a function $F$ in $d$ Boolean variables, where $F_i$ is the value of $F$ on inputs of Hamming weight $i$. $\mathbf{F}_{\text{ExtOne}}=\{[0,1,0,\cdots,0]\}$ is the set of functions called Exactly One function. $\mathbf{W}=\{[1,0,w]| w \in \mathcal{C}\} $ is the proper set of functions corresponding to complex number weighted edges. \#Pl-PerfMatch problem is \#Pl-$\mathbf{W}|\mathbf{F}_{\text{ExtOne}}$.  We use $\mathbf{M}$ to denote the set of the functions that can be realized by matchgates. \#Pl-$\mathbf{M}|\mathbf{M}$ can be reduced to \#Pl-PerfMatch\footnote{In a general \#BCSP problem, we shall be careful about the connection of gadgets, which respects the bipartition of the problem. Here, in \#Pl-$\mathbf{W}|\mathbf{F}_{\text{ExtOne}}$, we can realize a binary equality function, whose two dangling edges shall be connected to a left vertex and a right vertex respectively. We can use it to change the type of dangling edges, so we do not need to worry the connection, and define only one set $\mathbf{M}$.}, and it is the problem utilized in holographic algorithms.

The sign crossover function $C$ can be realized by matchgate \cite{Valiant08-HA,DBLP:journals/toc/CaiG14} . Figure \ref{Fig C} gives a matchgate whose function is $C$, where the default edge weight is $1$. For simplicity, we do not draw the vertex standing for the edge weight function. When all dangling edges take value $1$, among all $2^7$ assignments to internal edges, there is only on assignment such that all six $F_{\text{ExtOne}}$ give value $1$ and at the same time the weight function of the center edge gives $-1$. Hence, $C(1111)=-1$.

\begin{figure}[hbtp]
	\begin{center}
		\includegraphics[width=0.3\textwidth]{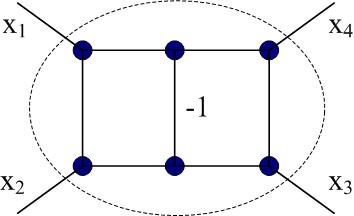}
	\caption{A matchgate realizing sign crossover function.}
	\label{Fig C}
	\end{center}
\end{figure}

A matchgate function can be written as a matrix. If an external edge is used as part of row (resp. column) index, we call it input (resp. output) edge. Recall the matrix in Definition \ref{def  C}. $x_1$ and $x_2$ are input edges, and $x_3$ and $x_4$ are output edges.

\subsection{Pfaffian}

Pfaffian  is a function defined for a skew-symmetric matrix which can be looked as a weighted graph $G(V,E,W)$.
Suppose $V=\{1,2,\ldots, 2n\}$. The sign of a perfect matching $\{(i_1,i_2),(i_3,i_4),\ldots,(i_{2n-1},i_{2n})\}$ in $G$ is either $1$ or $-1$ decided by the parity of permutation $(i_1, i_2, \cdots, i_{2n})$.
Pfaffian$(G)=\sum_{\pi=(i_1, i_2, \cdots, i_{2n})} \text{sgn}(\pi) W((i_1,i_2)) \cdots W((i_1,i_2))$.

The FKT algorithm computes \#Pl-PerfMatch by reducing it to Pfaffian.

Pfaffian can also be reduced to \#Pl-PerfMatch through the sign crossover function $C$.
If we arrange the vertices of $G$ on a circle, and draw the edges as lines with a tiny shake such that no 3 lines share a crossing,
the parity of the number of crossings between lines in $\pi$ is equal to the parity of $\pi$.
 Each vertex of $G$ is associated with an Exactly One function.
We change each crossing to a vertex of degree 4,
and then  replace each new vertex by one matchgate  $C$ to get a planar graph $G'(V',E',W')$.
\[\text{Pfaffian}(G)=2^n n! \#\text{Pl-PerfMatch}(G').\]

In $G'$, we add an external edge to each vertex in $V$, to get a matchgate denoted by $\Gamma_G$.

\begin{definition}
We call the above $\Gamma_G$ the matchgate introduced by $G$, and call $G$ the underlying graph of matchgate $\Gamma_G$.
\end{definition}

If the input $Z_{i,j}$ of  $\Gamma_G$ has only two $0$s and leaves only vertices $i$ and $j$ on the outer face unmatched by external edges, then $\Gamma_G(X)=W(i,j)$ by the definition and construction.
In Section \ref{sec ele}, we get the weights of a underlying graph by this equation.
The weight matrix of an underlying graph is an intuitive way to express the core part of an introduced matchgate function.

\subsection{The core of a matchgate}

Grassmann-Pl\"ucker identities are a set of polynomial identities about the Pfaffian of submatrices of $G$.
Grassmann-Pl\"ucker identities can be translated to identities about signature of matchgates, called matchgate identities \cite{Valiant02-QC,CC1,CCL2,DBLP:journals/toc/CaiG14}.

\begin{theorem}[\cite{CC1,CCL2,DBLP:journals/toc/CaiG14}] \label{prop:kehua1}
If $F$  is the signature of some matchgate, then $F$ satisfies all matchgate identities.
\end{theorem}

Define $U_k(X)=\{Y| \text{dist}(X,Y) \leq k\}$, where $X, Y \in \{0,1\}^n$, and $\text{dist}(X,Y)$ is the Hamming distance.

\begin{theorem}[\cite{CC1,CCL2,DBLP:journals/toc/CaiG14}] \label{thm: unique}
If a function $F$ is not $0$ on input $X$ and it  satisfies all matchgate identities, then $F$ is uniquely decided by its values on the set $U_2(X)$.
\end{theorem}

We call such a set $U_2(X)$ the core of the matchgate.
 The proof is based on  the observation that, for any $Y \in U_{k+2}(X)-U_{k}(X)$, there is a matchgate identity includes only one $F(Y)$ and the values of $F$ on $U_{k}(X)$. In this way, $U_2$ decides $U_4$, and $U_4$ decides $U_6$, and so on until $U_n$.

\begin{theorem}[\cite{CC1,CCL2,DBLP:journals/toc/CaiG14}]\label{thm:base graph}
Given a nonzero value $f_X$ and $n \choose 2$ arbitrary values $f_Y$, for each $Y \in U_2(X)$,
there is a matchgate $F$ such that $F(Y)=f_Y$ for all $Y \in U_2(X)$.
\end{theorem}

We explain how to prove a special case that $X={\bf 1}$ and $f_X=1$.
Let $Z_{i,j}$ denote the 0-1 string whose $i$th and $j$th bits are $1$ and all other bits are $0$.
$\{Z_{i,j}| i \in [n], j\in [n], i\neq j\}=U_2(\mathbf{1})$.
Define a weighted graph $G(V=[n],E,W)$. $W(i,j)=f_{Z_{i,j}}$. $G$ introduces a  matchgate $\Gamma_G$.
It is easy to check that $\Gamma_G(Z_{i,j})=W(i,j)=f_{Z_{i,j}}$.

\begin{theorem}[\cite{CC1,CCL2,DBLP:journals/toc/CaiG14}] \label{prop:kehua2}
If $H$ satisfies all matchgate identities, then $H$  is the signature of some matchgate.
\end{theorem}

When $H$ is not zero function, we apply Theorem \ref{thm:base graph} with the values of $H$ on a core, to get a matchgate $F$. By Theorem  \ref{thm: unique}, $F$ is the same as $H$.

\section{Elementary matchgate transformation} \label{sec ele}

Given a matchgate $\Gamma$ with external edges $e_1,e_2,\ldots, e_{s+t}$. Assume vertex $i$ is the endpoint of $e_i$, and they are all distinct w.l.o.g.. Its function $F$ has domain $[2]^{s+t}$. We also look $F$ as a matrix $(F_{I,J})$ with row index $I=e_1e_2\cdots e_{s} \in [2]^s$ and column index $J= e_{s+t}e_{s+t-1}\cdots e_{s+1} \in [2]^t$.

We define five kinds of matchgate transformations called flip, global factor, exchange, bar and slash. Each of them adds a small new part $\Delta$ connected to one or two of the dangling edges $e_1,e_2,\ldots, e_{s}$. The new matchgate $\Theta$ is a composition of two matchgates $\Gamma$ and $\Delta$.

All the transformations are invertible. Our purpose is to simplify the underlying graph to a canonical form. The first two transformations help to reach a matchgate admitting an underlying graph \cite{CC1}.
For the last three transformations, we need to analyze their affections on underlying graphs.

\subsection{Flip}

The flip applied to an external edge $e_i$ for some $i \in \{1,2,\ldots,s\}$, connects a new vertex $i'$ together with its dangling edge $e_{1'}$ to $e_i$.

\begin{figure}[hbtp]
	\begin{center}
		\includegraphics[width=0.3\textwidth]{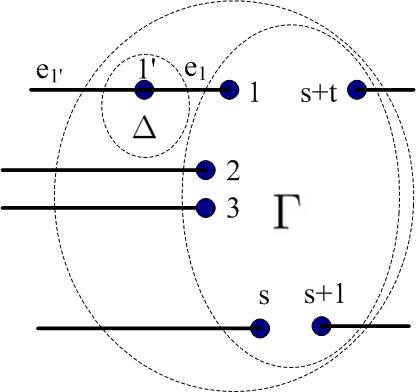}
	\caption{Matchgate $\Theta$.}
	\label{Fig flip}
	\end{center}
\end{figure}

For an example $i=1$, see Figure \ref{Fig flip}.

\[F_\Theta (e_{1'}=1-e_1, e_{2}, \ldots, e_{s+t})=F_\Gamma (e_1, e_2, \ldots, e_{s+t}).\]

Because $\left ( \begin{array}{cc} 0 & 1  \\  1  & 0 \\ \end{array}  \right ) $ is the function of $\Delta$, we have an equation in martix form \[F_\Theta=( \left ( \begin{array}{cc} 0 & 1  \\  1  & 0 \\ \end{array}  \right ) \otimes E_2^{\otimes (s-1)})  F_\Gamma.\]

It is not hard to get the following proposition.

\begin{proposition}
The inverse of a flip is itself. Flip keeps the rank of matchgate function.
\end{proposition}

\subsection{Global factor}

Global factor transforms $F_\Gamma$ to $cF_\Gamma$. We add two new nodes $n+1,n+2$ and a new edge $(n+1,n+2)$ of weight $c$ to $\Gamma$ to get a new matchgate $\Theta$. Obviously, the function of $\Theta$ is $c \Gamma$. Constant $c$ is always nonzero, when used in a global factor.
\begin{proposition}
The inverse of a global factor $c$ transformation is the global factor $1/c$ transformation. The global factor transformation keeps the rank of matchgate function.
\end{proposition}

\subsection{Getting the underlying graph}

Flip and global factor are used in the proof of Theorem \ref{thm:base graph}, to get the transformation between the general case and the special case that admits a underlying graph.

Assume $F_\Delta(I=e_1e_2\cdots e_s,J) \neq 0$. For each $e_i=0$, we apply a flip on $e_i$.
We get a new matchgate $\Theta_1$, and $F_{\Theta_1}(\mathbf{1},J)=F_\Delta(I,J)$.
We apply some proper flips to the external edges on the right side similarly, working as column transformations. We get a  matchgate  $\Theta_2$, and $F_{\Theta_2}(\mathbf{1},\mathbf{1})=F_\Delta(I,J)$.
Then, we apply a global factor $1/F_{\Theta_2}(\mathbf{1},\mathbf{1})$.

We get a matchgate $\Omega$ satisfying $F_\Omega(\mathbf{1})=1$. By the proof sketch for the special case of Theorem \ref{thm:base graph}, function $F_\Omega$ admits an underlying graph $G$.

We prove Theorem \ref{thm:base graph} for $f=F_\Delta$. $F_\Delta$ shows only its values on the core part. We look $\Delta$ as a black box. By the above invertible transformations, we get $A F_\Delta B= F_\Omega$.
Because $F_\Omega(\mathbf{1})=1$, we can construct a matchgate $\Gamma_G$ introduced by $G$, such that $\Gamma_G$ respects the values on the core part. By Theorem \ref{thm: unique}, $F_{\Gamma_G}=F_\Omega$. Now, we can solve the unknown $F_\Delta$ in the equation  $A F_\Delta B=F_{\Gamma_G}$. $A^{-1} \Gamma_G B^{-1}$ gives a construction for $F_\Delta$.


\subsection{Exchange}

An exchange applied to $e_i$ and $e_{i+1}$ of $\Gamma$, adds to $\Gamma$ a sign crossover matchgate applied on $e_{i'},e_{(i+1)'},e_{i+1},e_i$, where $e_{(i+1)'},e_{i'}$ are external edges of the new matchgate $\Theta$. An example of $i=1$ is shown in Figure \ref{Fig exchage}.

Recall the matrix form of sign crossover function $C$. We have

\[F_\Theta= (\left ( \begin{array}{cccc} 1 & 0 & 0 & 0  \\ 0 & 0 & 1 & 0 \\ 0 & 1 & 0 & 0  \\ 0 & 0 & 0 & -1  \end{array} \right ) \otimes E_2^{\otimes (s-2)})F_\Gamma.\]

\[F_\Theta (e_{1'}=e_2, e_{2'}=e_1, \ldots, e_{s+t})=C(e_2,e_1,e_2,e_1) F_\Gamma (e_1, e_2, \ldots, e_{s+t}).\]

\begin{figure}[hbtp]
	\begin{center}
		\includegraphics[width=0.35\textwidth]{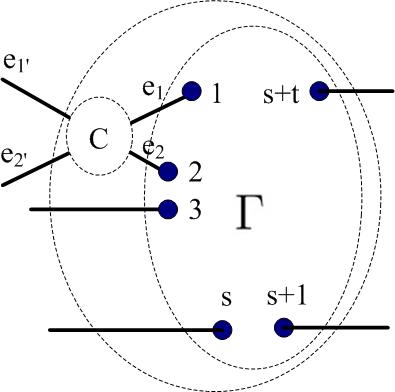}
	\caption{Matchgate $\Theta$.}
	\label{Fig exchage}
	\end{center}
\end{figure}

\begin{proposition}
The inverse of a $i,i+1$ exchange is itself. Exchange keeps the rank of matchgate function.
\end{proposition}

\subsection{Exchange and base graph}

Assume an exchange on $i,i+1$ is applied to a matchgate $\Gamma$ which admits an underlying graph.
Because $F_\Theta (\mathbf{1})=C(\mathbf{1}) F_\Gamma (\mathbf{1})=-F_\Gamma (\mathbf{1})$, we go on to apply a global factor $-1$ to $\Theta$ to get $\Omega$. Then $\Omega$ admits an underlying graph.

It's obviously, after a $i,i+1$ vertex rename, all corresponding edges in the underlying graphs of $\Omega$ and $\Gamma$ have weights of the same absolute value.

\subsection{Bar}

Suppose we have a matchgate $\Gamma$ and its underlying graph $G$. The edge $(i,i+1)$ has weight $w$.
A weight $-w$ bar applied to $e_i$ and $e_{i+1}$ of $\Gamma$ connects a matchgate $\Delta$ as shown in Figure \ref{Fig bar} to $e_i$ and $e_{i+1}$.

\begin{figure}[hbtp]
	\begin{center}
		\includegraphics[width=0.4\textwidth]{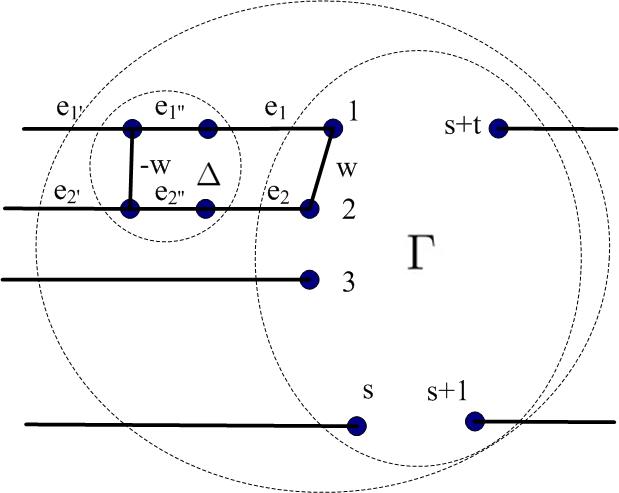}
	\caption{Matchgate $\Theta$.}
	\label{Fig bar}
	\end{center}
\end{figure}

By definition,
\[F_\Theta (e_{1'}, e_{2'}, e_3, \ldots, e_{s+t})=\sum_{e_1,e_2}F_\Delta(e_{2'},e_{1'},e_1,e_2) F_\Gamma (e_1, e_2, e_3, \ldots, e_{s+t}). \]

\[F_\Theta= (\left ( \begin{array}{cccc} 1 & 0 & 0 & -w  \\ 0 & 1 & 0 & 0 \\ 0 & 0 & 1 & 0  \\ 0 & 0 & 0 & 1  \end{array} \right ) \otimes E_2^{\otimes (s-2)})F_\Gamma.\]

\begin{proposition}
The inverse of a weight $-w$ bar applied to $e_i$ and $e_{i+1}$ is a weight $w$ bar applied to $e_i$ and $e_{i+1}$. Bar keeps the rank of matchgate function.
\end{proposition}

\subsection{Bar and base graph}

Because $\Delta$ forces $e_1=e_2=1$ with value $1$ when $e_{1'}=e_{2'}=1$, $F_{\Theta}(\mathbf{1})=1$ and $\Theta$ admits an underlying graph. We wonder the underlying graph $G'(V,E,W')$ of $\Theta$.

We observe that $\Delta$ forces $e_1=e_{1'}$ and $e_2=e_{2'}$ in most cases, except that $e_{1'}=e_{2'}=0$.

Hence, $W'(i,j)=W(i,j)$ holds for all $i,j \in [s+t]$ except $i=1,j=2$.

\[W'(1,2)=F_{\Theta}(0,0,1,\ldots,1)=\sum_{e_1,e_2}F_\Delta(0,0,e_2,e_1) F_\Gamma (e_1, e_2, 1,\ldots,1).\]

Check Figure \ref{Fig bar}, we find that when $e_{1'}=e_{2'}=0$, either we use $e_{1''}$ and $e_{2''}$ to match the left two vertices in $\Delta$ and $e_1=e_2=0$, or we use $-w$ to match them and $e_1=e_2=1$.  $W'(1,2)=1 \cdot w + (-w) \cdot 1=0$.

\begin{proposition}
A weight $-W(i,i+1)$ bar applied to $e_i$ and $e_{i+1}$ in a matchgate $\Gamma$ introduced by the underlying graph $G(V,E,W)$, gives a new matchgate $\Theta$ with almost the same underlying graph, except that $W(i,i+1)$ becomes $0$.
\end{proposition}

\subsection{Slash}

Suppose matchgate $\Gamma$ has a bipartite underlying graph $G([s],[t],E,W)$. W.l.o.g., we show a slash applied to $e_1$ and $e_2$. Suppose $W(1,s+t)=w_1 \neq 0$ and $W(2,s+t)=w_2$. This weight $-w_2/w_1$ slash connects a matchgate $\Delta$ as shown Figure \ref{Fig slash} to $\Gamma$ to construct a new matchgate $\Theta$.
The edge $(2',1'')$ has weight $-w_2/w_1$. The other two internal and four external edges has the default weight $1$.

\begin{figure}[hbtp]
	\begin{center}
		\includegraphics[width=0.4\textwidth]{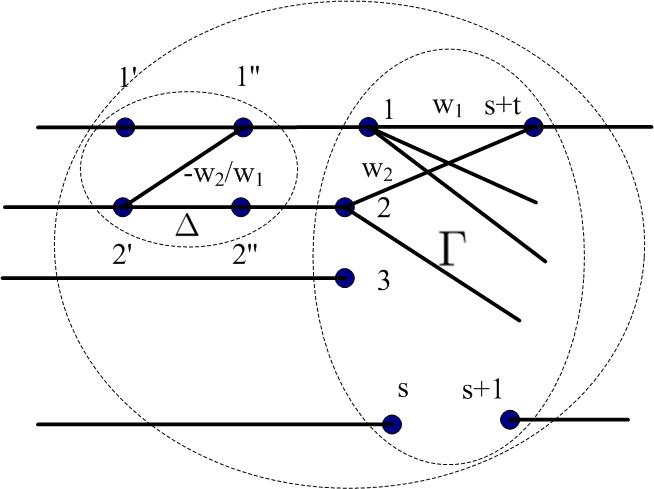}
	\caption{Matchgate $\Theta$.}
	\label{Fig slash}
	\end{center}
\end{figure}

By definition,
\[F_\Theta (e_{1'}, e_{2'}, \ldots, e_{s+t})=\sum_{e_1,e_2}F_\Delta(e_{1'},e_{2'},e_2,e_1) F_\Gamma (e_1, e_2, \ldots, e_{s+t}). \]

\[F_\Theta= (\left ( \begin{array}{cccc} 1 & 0 & 0 & 0  \\ 0 & 1 & 0 & 0 \\ 0 & -w_2/w_1 & 1 & 0  \\ 0 & 0 & 0 & 1  \end{array} \right ) \otimes E_2^{\otimes (s-2)})F_\Gamma.\]

\begin{proposition}
The inverse of a weight $-w$ slash applied to $e_i$ and $e_{i+1}$ is a weight $w$ slash applied to $e_i$ and $e_{i+1}$. Slash keeps the rank of matchgate function.
\end{proposition}

\subsection{Slash and base graph}

Slash can be analyzed similarly as bar, but more complicated, since it affects no only one edge weight in the underlying graph. We give a high level relationship about two multiplied bipartite underlying graphs and their multiplied matchgate functions, and use it to analyze slash. A bipartite underlying graph  $G([s],[t],E,W)$ can be expressed by its weight matrix $W=(W(i,j))$, $i \in [s]$, $j \in [t]$ in the standard way.

\begin{theorem} \label{thm: same product}
Suppose $G_1([s],[t],E_1)$ and  $G_1([t],[p],E_2)$  are two bipartite graphs with weight matrix $W_1$ and $W_2$ respectively. They introduce two matchgates $\Gamma_1$ and $\Gamma_2$, which define two functions  $F_{\Gamma_1}$ and $F_{\Gamma_2}$ respectively, as matrices of size $2^s \times 2^t$ and $2^t \times 2^p$. Then, the matchgate introduced by the bipartite graph with weight matrix $W_1W_2$, has function  $F_{\Gamma_1}F_{\Gamma_2}$.
\end{theorem}

\begin{proof}
We connect the $t$ output external edges of $\Gamma_1$ with the  $t$ input external edges of $\Gamma_2$  to get a new matchgate $\Theta$, whose function is $F_{\Gamma_1}F_{\Gamma_2}$.

Since the underlying graphs are bipartite, if input edges has $q$ $0$s, $F_{\Gamma_1}$ and $F_{\Gamma_2}$ force the output edges has also $q$ $0$s. When all external edges take value $1$, the new matchgate gives value $1$, so $\Theta$ admits a underlying graph.

After calculating the value of $F_\Theta$ on the core part, we find that its underlying graph is a bipartite graph and its weight matrix coincides with the matrix product $W_1W_2$.
\end{proof}

$\Gamma_1$ and $\Gamma_2$ communicate through $t$ edges. In the matchgate function form, the  $t$ edges send information in a length $t$ string with $2^t$ possibilities. In the bipartite graph form, we force one source in  $[s]$ and one sink in $[p]$, so the $t$ edges send information by a length $t$ string of Hamming weight $t-1$. There is a path in $W_1W_2$ from source to sink, is equal to, there is a perfect matching from source to sink in the new matchgate $\Theta$.

Since the strings of Hamming weight $t-1$ is a subset of the matchgate index set $[2]^t$, there is another intuitive way to prove Theorem \ref{thm: same product}. We illustrate it with the example $s=t=p=3$. $F_{\Gamma_1}$ is a square matrix indexed by $000, 001, 010, 011, 100, 101, 110, 111$.
Because $\Gamma_1$ has a bipartite underlying graph, $F_{\Gamma_1}$  has form

\[\left ( \begin{array}{cccccccc}
 \& & \& & \& & 0 & \& & 0 & 0 & 0 \\
 \& & \& & \& & 0 & \& & 0 & 0 & 0 \\
 \& & \& & \& & 0 & \& & 0 & 0 & 0 \\
 0 & 0 & 0 & * & 0 & * & * & 0 \\
 \& & \& & \& & 0 & \& & 0 & 0 & 0 \\
 0 & 0 & 0 & * & 0 & * & * & 0 \\
 0 & 0 & 0 & * & 0 & * & * & 0 \\
 0 & 0 & 0 & 0 & 0 & 0 & 0 &  1
 \end{array} \right ),  \]

where $*$ stands for a weight of underlying graph, and $\&$ stands for a function value not on the core.
Some value $0$ are also weights of the  underlying graph, but the others are not.

$F_{\Gamma_1}F_{\Gamma_2}$ is the product of two such matrices. It is block wise, so the $*$ part is an isolated multiplication. The product restricted to $\{111\}$ tells us the new matchgate admits underlying graph. The $0$s in last row and last column of the product tell us the underlying graph is bipartite. The product restricted block $\{011,101,110\}$ is $W_1W_2$ give use the rest of the conclusion.

\section{Canonical form}  \label{sec cnanonical form}

\begin{theorem} \label{thm canonical form}
Given a matchgate $\Gamma$ of $s$ input external edges and $t$ output external edges. The following holds.
\begin{enumerate}
\item
We can apply elementary matchgate transformations to $\Gamma$, to get a new matchgate $\Theta$, such that
$\Theta$ has the same function as a matchgate $\Omega$, where $\Omega$ has a weighted bipartite underlying graph $G([s],[t],E,W)$ and $E$ is a matching $M=\{(1,s+t), (2,s+t-1), \ldots, (r,s+t-r+1)\}$ of size $r$.
\item
Starting from $\Omega$, and applying the inverse of the above transformations, we get a matchgate construction for $F_\Gamma$.
\item $F_\Gamma$ has rank $2^r$, where $r \leq \min \{s,t\}$, depending on $\Gamma$.
\end{enumerate}
\end{theorem}

\begin{proof}
If $F_\Gamma$ is the zero function, $r=0$ and $\Theta$ is simply the empty matchgate.

Suppose $F_\Gamma$ is not zero function. We apply  elementary matchgate transformations to it.
\begin{enumerate}
\item Apply some flips to move a nonzero value $F_{\Gamma} (I,J)$.  We get a new matchgate $\Theta_1$ such that $F_{\Theta_1}= F_{\Gamma} (I,J) \neq 0$.

\item Apply global factor $1/F_{\Theta_1}$ to get $\Theta_2$. $\Theta_2$ admits a underlying graph $G_2([s+t],E,W_2)$.

\item  Whenever there is an edge $(i,j) \in [s] \times [s]$ with nonzero weight, apply exchange to move vertices $i,j$ to positions $1,2$ respectively, followed by a global factor $-1$ to keep the existing of an underlying graph if necessary. We get a new underlying graph $G_3([s+t],E,W_3)$. Apply weight $-W_3(1,2)$ bar to $e_1$ and $e_2$.

\item Do the similar thing to remove the edges in $[t] \times [t]$. We get a new bipartite underlying graph $G_4([s+t],E,W_4)$.

\item Apply a series of row slashes to $\Gamma_{G_4}$. By Theorem \ref{thm: same product}, the new matchgate is equal to an introduced matchgate with clear underlying wight matrix form. We pick the slashes, such that they transform the weight matrix $W_4$ almost into a upper triangular form, except the lack of a row reordering, since the weight matrix of slash corresponds to row addition transformation on weight matrix.

\item Apply a series of column slash to make the underlying graph into a matching.

\item Apply proper row and column exchanges and necessary global factor $-1$ to  get $G$.
\end{enumerate}

Because all these transformations are invertible, we may apply the inverse to the resulting matchgate, to get a construction of $\Gamma$.

Because $F_G$ has rank $2^r$, so is $F_\Gamma$.
\end{proof}

\section{Base collapse} \label{sec collapse}

Suppose $M$ is a matrix of size $n \times 2^t$ used as base in a holographic algorithm for $\#\mathbf{F}|M \mathbf{H}$.
Of course the rank $r'$ of $M$ is no more than $\min \{n, 2^t\}$. Base collapse occurs when $2^t >r'$.
The proof is separated into two parts. One part shows under some condition, how to get a new base matrix $N$ with many zero columns. The other part is the following lemma.

\begin{lemma} \label{lem 0columns}
Suppose the base $M$  used in a holographic algorithm for $\#\mathbf{F}|M \mathbf{H}$,  satisfies that if $M(x,y_1y_2\cdots y_t) \neq 0$ then $y_{r+1}=y_{r+2}=\cdots =y_t=1$.  Then the base $N$, where $N(x,y_1\cdots y_r)=M(x,y_1y_2\cdots y_r 1 \cdots 1)$ can be used to give a holographic algorithm for $\#\mathbf{F}|M \mathbf{H}$ too. \footnote{
The condition in this lemma can be relaxed to, that there exist constants $c_{r+1},c_{r+2},\ldots,c_l \in [2]$ such that if $M(x,y_1y_2\cdots y_l) \neq 0$ then $y_{r+1}=c_{r+1}, y_{r+2}=c_{r+2},\ldots y_l=c_l$.
}
\end{lemma}

\begin{proof}
The holographic algorithm for $\#\mathbf{F}|M \mathbf{H}$ computes it as a  $\#\mathbf{F}M|\mathbf{H}$ problem, as shown in the left picture of Figure \ref{Fig 0columns}. For each $F \in \mathbf{F}$ or arity $r_F$, $F  M^{\otimes r_F}$ is a matchgate (generator). Each $H \in \mathbf{H}$ is a matchgate (recognizer).
Generators and recognizers are connected by bunches of edges. Each bunch contains $t$ edges, for example $e_1, e_2, \ldots, e_t$.

\begin{figure}[hbtp]
	\begin{center}
		\includegraphics[width=0.8\textwidth]{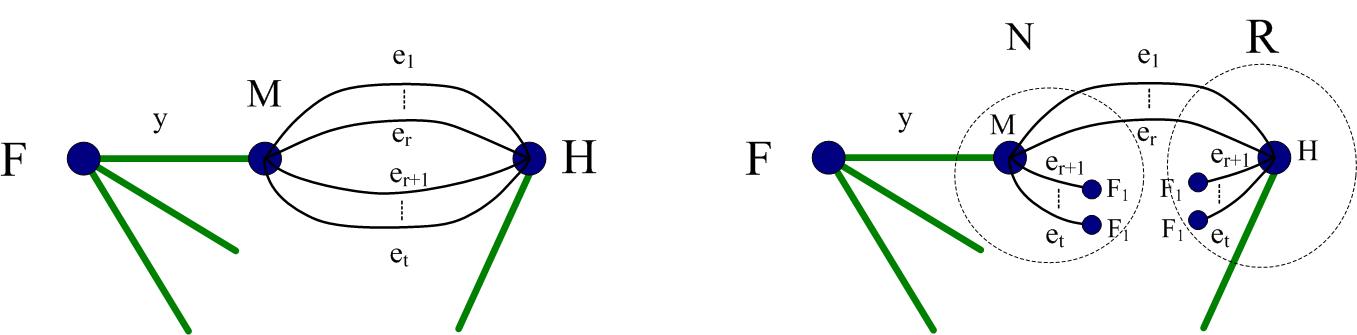}
	\caption{Holographic algorithms use base $M$ and $N$ respectively.}
	\label{Fig 0columns}
	\end{center}
\end{figure}

Given any instance of $\#\mathbf{F}M|\mathbf{H}$, the value is a summation over all assignments of edges, including $e_1, \ldots, e_t$. By the condition, if one of edges $e_{r+1}, \ldots, e_t$ is assigned value $0$, $M$ contributes a value $0$ as one factor. In this case, $F  M^{\otimes r_F}$ contributes $0$ too.

Hence, we only need to sum over the assignments assigning $0$ to $e_{r+1}, \ldots, e_t$ in each bunch. Then, the value is equal to the value of the second graph in Figure \ref{Fig 0columns}, where we cut $e_{r+1}, \ldots, e_t$ and connect onto them the unary Exactly One function $F_1=[0,1]$.  $F_1$ forces its input to be $1$.

We look $M$ connected by $t-r$ $F_1$ functions as the new base matrix $N$. Each $H \in \mathbf{H}$ becomes $R$, which is a function $H$ connected by $r_H(t-r)$ many $F_1$ functions. Function set $\mathbf{H}$ becomes $\mathbf{R}$.

We already explained that we get the same value as the original instance.
By the same reason, $N\mathbf{R}=M\mathbf{H}$.
Because we use only Exactly One function to modify the generators and recognizers, $F  N^{\otimes r_F}$ and $R$ are still matchgate. Hence, we get a holographic algorithm  $\#\mathbf{F}N|\mathbf{R}$ for  $\#\mathbf{F}|N\mathbf{R}$ which is also  $\#\mathbf{F}|M \mathbf{H}$, using a collapsed base $N$.
\end{proof}

Given a base $M$ of size $n \times 2^t$, we want to get a new base $N$ of the same size but with many zero columns. This can be achieved by common column transformation of matrix, but to keep $F M^{\otimes s}$ a matchgate, only  matchgate transformations shall be used.

\begin{definition}
Suppose $M$ is a matrix of rank $2^r$. A matrix $A$ is called a full rank matchgate realizer of $M$, if $AM$ is a matchgate of rank $2^r$.
\end{definition}

Because the rows of $AM$ and $M$ have the same rank, there exists a general inverse $A^{-1}$ of $A$ such that $A^{-1} A M=M$. We apply elementary row transformations $B$ and column transformations $C$  to $AM$ to get its canonical form $BAMC$ by Theorem \ref{thm canonical form}, which satisfies that if $BAMC(x,y_1y_2\cdots y_t) \neq 0$ then $y_{r+1}=y_{r+2}=\cdots =y_t=1$.
Both $B$ and $C$ have their inverse $B^{-1}$ and $C^{-1}$.
$MC=A^{-1}B^{-1}BAMC$ satisfies the same condition,  since the row transformation $A^{-1}B^{-1}$ do not affect zero columns. We emphasize that  $A^{-1}B^{-1}$  is not necessary a matchgate transformation.

It is obvious that besides $\#\mathbf{F}M|\mathbf{H}$, $\#\mathbf{F}(MC)|C^{-1}\mathbf{H}$ is also a holographic algorithm  for  $\#\mathbf{F}|M \mathbf{H}$. We have proved the following lemma.

\begin{lemma} \label{lem base MC}
Suppose the base $M$  of a holographic algorithm has size $n \times 2^t$ and rank $2^r$. Suppose $M$ has a full rank matchgate realizer. Then there exists a series of elementary matchgate column transformations, which as a union can be expressed as a matrix $C$ of size $2^t \times 2^t$, such that there is a holographic algorithm for the same problem using base $MC$. What's more, $MC$ satisfies that if $MC(x,y_1y_2\cdots y_t) \neq 0$ then $y_{r+1}=y_{r+2}=\cdots =y_t=1$.
\end{lemma}

Put Lemma \ref{lem base MC} and Lemma \ref{lem 0columns} together directly, we get a base collapse theorem.

\begin{theorem} \label{thm basecollapse1}
Suppose the base $M$  of a holographic algorithm has size $n \times 2^t$ and rank $2^r$. Suppose $M$ has a full rank matchgate realizer. Then there is a holographic algorithm solving the same problem with a base of size $n \times 2^r$.
\end{theorem}

Lemma \ref{lem base MC} and Theorem \ref{thm basecollapse1} require the existence of a full rank matchgate realizer. In fact, what we need is a matchgate function $P$ such that each row of $M$ can be linearly expressed by the rows of $P$. That is, there exists $Q$ such that $QP=M$. Then we simplify $P$ to $BPC$, and use $QB^{-1}BPC=MC$ as the new base.
We call $P$ a matchgate cover of $M$.

\begin{theorem} \label{thm basecollapse new added}
Suppose the base $M$  of a holographic algorithm has size $n \times 2^t$. Suppose $M$ has a matchgate cover of rank $2^r$. Then there is a holographic algorithm solving the same problem with a base of size $n \times 2^r$.
\end{theorem}

The ideal condition is nothing but that there is a holographic algorithm using base $M$. Then there is a matchgate $F M^{\otimes s}$. If we write $F$ into a size $n^{s-1} \times s$ matrix,\footnote{There are several ways to get such a matrix, if $F$ is not symmetric.} the matchgate also has expression $M' \ ^{\otimes s-1} F M$. We hope one of these candidates $M'\ ^{\otimes s-1} F$ is a full rank matchgate realizer.
If $F M^{\otimes s}$ gives a full rank matchgate realizer, it is called a generator of full rank in \cite{CaiFu}.
Unfortunately, the existence of a  holographic algorithm does not obviously promise a generator of full rank, or a full rank matchgate realizer, or a matchgate cover of small rank.

We introduce another base collapse theorem using another condition.

\begin{theorem}  \label{thm final}
Suppose there is a holographic algorithm $\#\{FM^{\otimes s}\}|\mathbf{H}$ for $\#\{F\}|M\mathbf{H}$, where $F$ is a symmetric function and $M$ is a base of size $n \times 2^t$. Then there is a holographic algorithm solving a problem equivalent to $\#\{F\}|M\mathbf{H}$, using a base $N$ of size $n \times 2^{r}$, where $r\leq \log n$.
\end{theorem}

\begin{proof}
By Theorem \ref{thm canonical form}, the rank of $M'\ ^{\otimes s-1} F M$ is a power of $2$. Obviously, the rank is no more than $n$. Suppose it is $2^{r}$.

Let $A$ denote $M'\ ^{\otimes s-1} F$. To utilize Lemma \ref{lem base MC}, we hope $A$ is a full rank matchgate realizer of $M$. Generally, this is not right. However, noticing $M$ serves only one function $F$ on the left side, we may get rid of the redundant rank in $M$ to achieve this.

The matrix $AM$ has $2^{t(s-1)}$ rows and rank $2^{r}$.  Suppose matrix $S_{2^{r}, 2^{t(s-1)}}$ selects a maximum independent row vector set $SAM$ of $AM$.

$\langle M \rangle$ denotes the space spanned by rows in $M$.
$\langle SAM \rangle = \langle AM \rangle \subseteq \langle M \rangle$.

There is a matrix $T_{n-2^{r},2^t}$, such that the block matrix $B=\left (  \begin{array}{c}  SAM \\ T  \end{array}\right )$ satisfying $\langle B \rangle = \langle M \rangle$ and $\langle M \rangle = \langle SAM \rangle \oplus \langle T \rangle$.

There exists an invertible matrix $C_{n \times n}$ such that $C B=M$. Separate $C$ into two blocks $\left (  \begin{array}{cc}  C_1 C_2  \end{array}\right )$ of size $n \times 2^{r}$  and  $n \times (n-2^{r})$.

Because $\langle M \rangle$ is the direct sum of  $\langle SAM \rangle$ and  $\langle T \rangle$,
$AM=ACB=AC_1 SAM+AC_2 T=AC_1 SAM$. Notice $F$ is symmetric and $A=M'\ ^{\otimes s-1} F$, $FM^{\otimes s}=F  (C_1 SAM) ^{\otimes s} = F  C_1^{\otimes s} (SAM) ^{\otimes s}$.

$\#\{FM^{\otimes s}\}|\mathbf{H}$  is equal to $\#\{ F  C_1^{\otimes s} (SAM) ^{\otimes s}\}|\mathbf{H}$, which is a holographic algorithm using base $SAM$ for  $\#\{ F  C_1^{\otimes s}\}| (SAM)\mathbf{H}$.

 $\#\{ F  C_1^{\otimes s}\}| (SAM)\mathbf{H}$ is equivalent to $\#\{F\}|M\mathbf{H}$. Its base $SAM$ has size $2^{r} \times 2^t$ and rank $2^{r}$. This base has a full rank matchgate realizer $AC_1$.  By Lemma \ref{lem base MC}, we get a new base, and then by Lemma \ref{lem 0columns}, the new base collapses to size $2^{r} \times 2^{r}$.
\end{proof}

From the proof of Theorem \ref{thm final}, we know the function $F$ utilizes only a subspace of $\langle M \rangle$.
In a general problem $\#\mathbf{F}|\mathbf{H}$, there are functions $F_1, F_2 \in \mathbf{F}$. Assume $F_1$ utilizes only a proper subspace of $\langle M \rangle$. If we cut off the other part of $\langle M \rangle$ as in the proof, then maybe we lost the part useful for $F_2$. This is the reason that Theorem \ref{thm final} requires a singleton set $\mathbf{F}$.

In this paper, we solve the most general base collapse of holographic algorithm under some conditions.
We give an example satisfying neither of the two conditions in Theorem \ref{thm basecollapse1} and \ref{thm final}. Suppose $\#\{F_1, F_2\}|M\mathbf{H}$ has a holographic algorithm using base $M$. Neither generator $F_1 M^{\otimes r_1}$ nor generator $F_2 M^{\otimes r_2}$ contains a full rank matchgate realizers for $M$. It is still possible to utilize a matchgate cover of $M$ to get some base collapse. However, there is no characterization of matchgate cover, and we do not know to which size the base can be collapsed.

\bibliographystyle{plain}

\bibliography{draft}

\begin{thebibliography}{10}

\bibitem{CC1}
Jin{-}Yi Cai and Vinay Choudhary.
\newblock On the theory of matchgate computations.
\newblock {\em Electronic Colloquium on Computational Complexity {(ECCC)}},
  (018), 2006.

\bibitem{DBLP:journals/ijsi/CaiC07}
Jin{-}Yi Cai and Vinay Choudhary.
\newblock Some results on matchgates and holographic algorithms.
\newblock {\em Int. J. Software and Informatics}, 1(1):3--36, 2007.

\bibitem{CCL2}
Jin{-}Yi Cai, Vinay Choudhary, and Pinyan Lu.
\newblock On the theory of matchgate computations.
\newblock {\em Theory Comput. Syst.}, 45(1):108--132, 2009.

\bibitem{CaiFu}
Jin{-}Yi Cai and Zhiguo Fu.
\newblock A collapse theorem for holographic algorithms with matchgates on
  domain size at most 4.
\newblock {\em Inf. Comput.}, 239:149--169, 2014.

\bibitem{DBLP:journals/corr/CaiFGW15}
Jin{-}Yi Cai, Zhiguo Fu, Heng Guo, and Tyson Williams.
\newblock A holant dichotomy: Is the {FKT} algorithm universal?
\newblock {\em CoRR}, abs/1505.02993, 2015.

\bibitem{DBLP:journals/toc/CaiG14}
Jin{-}Yi Cai and Aaron Gorenstein.
\newblock Matchgates revisited.
\newblock {\em Theory of Computing}, 10:167--197, 2014.

\bibitem{DBLP:conf/coco/CaiL07}
Jin{-}Yi Cai and Pinyan Lu.
\newblock Bases collapse in holographic algorithms.
\newblock In {\em 22nd Annual {IEEE} Conference on Computational Complexity
  {(CCC} 2007), 13-16 June 2007, San Diego, California, {USA}}, pages 292--304.
  {IEEE} Computer Society, 2007.

\bibitem{DBLP:journals/tcs/CaiL09}
Jin{-}Yi Cai and Pinyan Lu.
\newblock Holographic algorithms: The power of dimensionality resolved.
\newblock {\em Theor. Comput. Sci.}, 410(18):1618--1628, 2009.

\bibitem{DBLP:journals/tcs/CaiL10}
Jin{-}Yi Cai and Pinyan Lu.
\newblock On blockwise symmetric signatures for matchgates.
\newblock {\em Theor. Comput. Sci.}, 411(4-5):739--750, 2010.

\bibitem{DBLP:journals/mst/CaiL10}
Jin{-}Yi Cai and Pinyan Lu.
\newblock On symmetric signatures in holographic algorithms.
\newblock {\em Theory Comput. Syst.}, 46(3):398--415, 2010.

\bibitem{DBLP:journals/jcss/CaiL11}
Jin{-}Yi Cai and Pinyan Lu.
\newblock Holographic algorithms: From art to science.
\newblock {\em J. Comput. Syst. Sci.}, 77(1):41--61, 2011.

\bibitem{DBLP:conf/focs/CaiLX08}
Jin{-}Yi Cai, Pinyan Lu, and Mingji Xia.
\newblock Holographic algorithms by fibonacci gates and holographic reductions
  for hardness.
\newblock In {\em 49th Annual {IEEE} Symposium on Foundations of Computer
  Science, {FOCS} 2008, October 25-28, 2008, Philadelphia, PA, {USA}}, pages
  644--653, 2008.

\bibitem{DBLP:conf/focs/CaiLX10}
Jin{-}Yi Cai, Pinyan Lu, and Mingji Xia.
\newblock Holographic algorithms with matchgates capture precisely tractable
  planar {\#}{CSP}.
\newblock In {\em 51th Annual {IEEE} Symposium on Foundations of Computer
  Science, {FOCS} 2010, October 23-26, 2010, Las Vegas, Nevada, {USA}}, pages
  427--436. {IEEE} Computer Society, 2010.

\bibitem{DBLP:journals/siamcomp/CaiLX11}
Jin{-}Yi Cai, Pinyan Lu, and Mingji Xia.
\newblock Computational complexity of {H}olant problems.
\newblock {\em {SIAM} J. Comput.}, 40(4):1101--1132, 2011.

\bibitem{DBLP:conf/soda/CaiLX11}
Jin{-}Yi Cai, Pinyan Lu, and Mingji Xia.
\newblock Dichotomy for {H}olant* problems of {B}oolean domain.
\newblock In {\em Proceedings of the Twenty-Second Annual {ACM-SIAM} Symposium
  on Discrete Algorithms, {SODA} 2011, San Francisco, California, USA, January
  23-25, 2011}, pages 1714--1728, 2011.

\bibitem{DBLP:conf/soda/CaiLX13}
Jin{-}Yi Cai, Pinyan Lu, and Mingji Xia.
\newblock Dichotomy for {H}olant* problems with domain size 3.
\newblock In {\em Proceedings of the Twenty-Fourth Annual {ACM-SIAM} Symposium
  on Discrete Algorithms, {SODA} 2013, New Orleans, Louisiana, USA, January
  6-8, 2013}, pages 1278--1295, 2013.

\bibitem{Kasteleyn1961}
Pieter~W. Kasteleyn.
\newblock The statistics of dimers on a lattice: I. {T}he number of dimer
  arrangements on a quadratic lattice.
\newblock {\em Physica}, 27(12):1209 -- 1225, 1961.

\bibitem{Kasteleyn1967}
Pieter~W. Kasteleyn.
\newblock Graph {T}heory and {C}rystal {P}hysics.
\newblock In {\em Graph Theory and Theoretical Physics}, pages 43--110.
  Academic Press, 1967.

\bibitem{LX08stacs}
Angsheng Li and Mingji Xia.
\newblock A theory for valiant's matchcircuits (extended abstract).
\newblock In {\em STACS}, pages 491--502, 2008.

\bibitem{TemperleyFisher1961}
H.~N.~V. {Temperley} and M.~{Fisher}.
\newblock {Dimer problem in statistical mechanics-an exact result}.
\newblock {\em Philosophical Magazine}, 6:1061--1063, August 1961.

\bibitem{Valiant02-QC}
Leslie~G. Valiant.
\newblock Quantum circuits that can be simulated classically in polynomial
  time.
\newblock {\em {SIAM} J. Comput.}, 31(4):1229--1254, 2002.

\bibitem{DBLP:conf/icalp/Valiant05}
Leslie~G. Valiant.
\newblock Holographic circuits.
\newblock In {\em Automata, Languages and Programming, 32nd International
  Colloquium, {ICALP} 2005, Lisbon, Portugal, July 11-15, 2005, Proceedings},
  pages 1--15, 2005.

\bibitem{DBLP:conf/focs/Valiant06}
Leslie~G. Valiant.
\newblock Accidental algorithms.
\newblock In {\em 47th Annual {IEEE} Symposium on Foundations of Computer
  Science {(FOCS} 2006), 21-24 October 2006, Berkeley, California, USA,
  Proceedings}, pages 509--517. {IEEE} Computer Society, 2006.

\bibitem{Valiant08-HA}
Leslie~G. Valiant.
\newblock Holographic algorithms.
\newblock {\em {SIAM} J. Comput.}, 37(5):1565--1594, 2008.

\bibitem{DBLP:conf/icalp/Xia10}
Mingji Xia.
\newblock Holographic reduction: {A} domain changed application and its partial
  converse theorems.
\newblock In {\em Automata, Languages and Programming, 37th International
  Colloquium, {ICALP} 2010, Bordeaux, France, July 6-10, 2010, Proceedings,
  Part {I}}, pages 666--677, 2010.

\bibitem{DBLP:journals/tcs/XiaZZ07}
Mingji Xia, Peng Zhang, and Wenbo Zhao.
\newblock Computational complexity of counting problems on 3-regular planar
  graphs.
\newblock {\em Theor. Comput. Sci.}, 384(1):111--125, 2007.

\end{thebibliography}

\end{document}